\documentclass[twocolumn]{autart}    % Enable this line and disable the
\usepackage{graphicx}    
\usepackage{color}   % Include this line if your
\usepackage{graphics} % for pdf, bitmapped graphics files
\usepackage{epsfig} % for postscript graphics files
\usepackage{mathptmx} % assumes new font selection scheme installed
\usepackage{times} % assumes new font selection scheme installed
\usepackage{amsmath} % assumes amsmath package installed
\usepackage{amssymb}  % assumes amsmath package installed
\usepackage{floatrow}
\usepackage{enumerate}
\usepackage{enumitem}
\usepackage{mathrsfs}
\usepackage{amsfonts}
\newenvironment{proof}{\par{\itshape Proof}.\ }
\newtheorem{theorem}{\indent Theorem}
\newtheorem{lemma}[theorem]{\indent Lemma}

\newtheorem{proposition}[theorem]{\indent Proposition}

\newtheorem{remark}[theorem]{\indent Remark}

\allowdisplaybreaks

\begin{document}

\begin{frontmatter}

\title{Two-step feedback preparation of entanglement for qubit systems with time delay\thanksref{footnoteinfo}} % Title, preferably not more

\thanks[footnoteinfo]{This work was supported in part by the Australian Research Council's
Discovery Projects Funding Scheme under Projects DP190101566, DP180101805, by the Air Force Office
of Scientific Research under Agreement FA2386-16-1-4065, by the
Centres of Excellence CE170100012, by the U.S. office of Naval Research Global under Grant N62909-19-1-2129, and by the National Natural Science Foundation under Grants Nos. 61828303, 61873251, and 61833010.}
\author[ADFA,CQC2T]{Yanan Liu}\ead{yaananliu@gmail.com},
\author[ADFA]{Daoyi Dong}\ead{dayidong@gmail.com},
\author[USTC]{Sen Kuang}\ead{skuang@ustc.edu.cn},
\author[ANU]{Ian R. Petersen}\ead{i.r.petersen@gmail.com},
\author[ADFA,CQC2T]{Hidehiro Yonezawa}\ead{h.yonezawa@adfa.edu.au}
\thanks[b]{Tel. +61-2-62686285, Fax +61-2-62688443 (Daoyi Dong).}

\address[ADFA]{School of Engineering and Information Technology, University of New South Wales, Canberra, ACT 2600, Australia}  % Please supply

\address[CQC2T]{Center for Quantum Computation and Communication Technology, Australian Research Council, Canberra, ACT 2600, Australia}

\address[USTC]{Department of Automation, University of Science and Technology of China, Hefei 230027, PR China}

\address[ANU]{Research School of Electrical, Energy and Materials Engineering, The Australian National University, Canberra, ACT 2601, Australia}

                             % keyword list or with the
                                          % help of the Automatica
                                          % keyword wizard

\begin{abstract}                          % Abstract of not more than 200 words.
Quantum entanglement plays a fundamental role in quantum computation and quantum communication. Feedback control has been widely used in stochastic quantum systems to generate given entangled states since it has good robustness, where the time required to compute filter states and conduct filter-based control usually cannot be ignored in many practical applications. This paper designed two control strategies based on the Lyapunov method to prepare a class of entangled states for qubit systems with a constant delay time. The first one is bang-bang-like control strategy, which has a simple form with swtching between a constant value and zero, the stability of which is proved. Another control strategy is switching Lyapunov control, where a constant delay time is introduced in the filter-based feedback control law to compensate for the computation time. Numerical results on a two-qubit system illustrate the effectiveness of these two proposed control strategies.
\end{abstract}

\end{frontmatter}

\section{Introduction}

Quantum entanglement is regarded as the most remarkable feature that distinguishes quantum physics from classical physics \cite{R1}. Since some entangled states were discovered in physical experiments, they have been extensively studied in quantum inseparability and investigated as a basic resource of quantum technologies \cite{R2}. Entanglement is the basis of nearly all quantum information protocols such as quantum teleportation \cite{R4}, quantum cryptography \cite{R5}, quantum spectroscopy with an atomic resolution \cite{R7} and quantum computation \cite{R6}. Bell states are regarded as the maximally entangled states in two-qubit systems \cite{R10}. GHZ entangled states are an extension of Bell states to multi-qubit systems. The preparation of multi-qubit entangled states is a fundamentally important task in many quantum technologies such as one-way quantum computation \cite{R26} and long-distance quantum communication \cite{R27}.

Compared with open-loop control strategies, feedback control usually has good robustness since it uses feedback information, such as filter states, to suppress the influence of noises and uncertainties \cite{Liu2019feedback,Gao2016fault}. It turns out that the quantum feedback control can be used to increase the entanglement of atoms \cite{R40},\\\cite{R41}, and to generate entanglement in superconducting quantum systems \cite{shankar2013autonomously},\\\cite{riste2013deterministic}. One of the main challenges in measurement-based feedback control comes from the presence of time delay in the feedback loop. The time required to compute filter states and conduct feedback control often cannot be ignored because of very fast quantum evolution \cite{Liu2019feedback}, which becomes a source of instability in quantum control systems and influences the control performance in experiments \cite{R29,R9}. There are two main solutions to deal with this issue. One solution is to reduce the computing time by decreasing the dimension of the filter equation \cite{R17,R18}. A feedback control law based on a reduced order filter equation can be designed to realise real-time feedback control. However, the inevitable error between the full order filter and the reduced-dimensional filter may cause filter-reduction-based control to be ineffective. Therefore, a second method has been proposed, where a constant delay time is compensated for in the control input. Switching control strategies with a compensation time have been proposed to achieve the preparation of eigenstates in two-dimensional and $N$-dimensional systems \cite{R19,R22}, and to prepare the entangled states in three-qubit systems \cite{R20}.

In this paper, we mainly focus on the preparation of a class of entangled states in multi-qubit systems (including two-qubit systems) with a constant delay time. The contributions of this paper are as follows. Firstly, we propose two switching control strategies, both of them can deal with the stabilisation problem with degenerate measured observables (a degenerate operator means that this operator has at least two identical eigenvalues). In this paper, we use two control channels simultaneously and design a switching control law by using the Lyapunov method, which overcomes the obstacle caused by degenerate observables. Secondly, this paper considers delay time in the feedback loop. A bang-bang-like control strategy with a simple form and a switching Lyapunov control strategy which has the potential to speed up convergence are designed to deal with the delay time. We prove the stability of the stochastic feedback control system with a constant delay time, and the system state converges to the target state with probability $1$. We also show that the proposed methods have robust performance with respect to imperfect measurement efficiency and dephasing noise.

The rest of this paper is organised as follows. A stochastic quantum system model is given in Section \ref{sec2}, where we also discuss a possible form of the observable to be measured. In Section \ref{sec3}, a bang-bang-like control law is designed based on the Lyapunov method and the division of the state space, where we construct control Hamiltonians and prove the stability of the time-delay system. Section \ref{sec4} proposes another control strategy with switching between a constant value and the time-delay Lyapunov feedback control law. Numerical results on a two-qubit system are presented in Section \ref{sec5} showing the effectiveness of the proposed feedback control strategies. Section \ref{sec6} includes concluding remarks.

%--------------------------------------------------------------------------------------------------
\section{System model and control task}
\label{sec2}
In this paper, we consider the generation of a given entangled state in an $N$-qubit system with time delay. By introducing homodyne measurement, one can obtain an optimal state estimate based on quantum filter equations. The feedback control law is then designed with the filter state. Such a feedback control system can be described by the following stochastic master equation (SME) \cite{R30}
\begin{equation}\label{eq7}
\left\{
\begin{aligned}
d\rho_t &=-i[H,\rho_t ]dt+\Gamma_A \mathcal{D}[A]\rho_t dt+\sqrt{\eta_A \Gamma_A } \mathcal{H}[A]\rho_t dW_t, \\
dy_t &=dW_t+\sqrt{\eta_A \Gamma_A } {\rm Tr}\left[(A+A^\dagger) \rho_t\right]dt. \\
\end{aligned}
\right.
\end{equation}
Here, $\rho_t$ is the density operator, representing the system state. $H$ is the system Hamiltonian, and can be written as $H\triangleq H_0+\sum_{k}u^{(k)}(t-\tau) H_k$, with the free Hamiltonian $H_0$ and the control Hamiltonians $H_k$. $H_0$ usually can be written as a diagonal matrix $H_0\triangleq {\rm diag}\left[h_1,h_2,\cdots,h_n\right]$, where $n$ is the dimension of state space and $n=2^N$ for $N$-qubit systems. $u^{(k)}(t-\tau)$ is control input with delay time $\tau$. $dW_t$ represents the noisy process caused by measurement and can be expressed as a standard Wiener process. $A$ is an observable described by a Hermitian matrix. $\Gamma_A$ and $\eta_A$ are the measurement strength and efficiency, respectively. $y_t$ is an observable process. $\mathcal{D}[A] \rho_t$ and $\mathcal{H}[A] \rho_t$ have the following forms:
\begin{equation}
\label{eq8}
\left\{
\begin{aligned}
\mathcal{D}[A] \rho_t &=A\rho_t A^\dagger -1/2(A^\dagger A\rho_t+\rho_t A^\dagger A), \\
\mathcal{H}[A] \rho_t &=A\rho_t+\rho_t A^\dagger -{\rm Tr}[(A+A^\dagger)\rho_t]\rho_t. \\
\end{aligned}
\right.
\end{equation}
The target state in this paper is a class of entangled states that has the following form \cite{R32}
\begin{equation}
\label{eq10}
|\varphi\rangle=\frac{1}{\sqrt{2}}\bigg(\prod_{j=1}^{N} |\zeta_j\rangle \pm \prod_{j=1}^N |\zeta_j^c \rangle\bigg),
\end{equation}
where $\zeta_j$ denotes a binary variable in the set $\{0,1\}$ and $\zeta_j^c=1-\zeta_j$, The system state can be written as a density operator $\rho=|\varphi \rangle \langle \varphi|$.

Our control objective is to prepare any given GHZ entangled state of the form (3). In order to achieve this objective, the observable operator $A$ needs to be chosen based on the specific target state. For ease of exposition, we assume that the target state is:
$$\rho_d=\frac{1}{2}(|0 \cdots 0\rangle+|1 \cdots 1\rangle)(\langle 0 \cdots 0|+\langle 1 \cdots 1|) =\frac{1}{2}\left[\begin{smallmatrix}
1 & 0&  \cdots &1\\
0 & 0&\cdots &0\\
\vdots & \vdots &\ddots &\vdots \\
1 & 0 &\cdots &1
\end{smallmatrix}\right].$$ 
According to quantum state reduction under continuous measurement \cite{R15}, we choose an observable operator such that the target state is one of its eigenstates; i.e., $A\rho_d=\lambda_d \rho_d$, where $\lambda_d$ is the eigenvalue related to $\rho_d$. For simplicity, we choose $A$ as a diagonal matrix:
\begin{equation}\label{eq9}
A={\rm diag}\{\lambda_d, \lambda_2, \cdots, \lambda_i, \cdots,  \lambda_d\}
\end{equation}
where $\lambda_i \neq \lambda_d,i=2,\cdots,n-1$. This kind of observable can be factorised as $A=a_1\sigma_z^{(1)}\otimes I_2^{(2)}\otimes \cdots \otimes I_2^{(N)}+ \cdots+a_N I_2^{(1)}\otimes \cdots \otimes \sigma_z^{(N)}$, where $I_2^{(i)}$ is the two-dimensional identity matrix and $\sigma_z^{(i)}$ is the Pauli operator $\sigma_z^{(i)}={\rm diag}\{1, -1\}$ and $a_1+\dots +a_N=0$. In this paper we focus on the systems with degenerate observables in \eqref{eq9}.%--------------------------------------------------------------------------------------------------
\section{Bang-bang like control with time delay}
\label{sec3}
\subsection{Control design}
In the case with degenerate observables, if only one control channel is used, it is usually difficult to design a feedback control law to achieve the preparation of the target state even when we choose a complicated Lyapunov function such as $V(\rho)=1-{\rm Tr}(\rho \rho_d )+c\left[{\rm Tr}(A^2 \rho)-{\rm Tr}^2 (A \rho)\right]$, where $c>0$ is a real constant \cite{R23}. Moreover, using only one control channel will impose more restrictions on the free Hamiltonian $H_0$ and the control Hamiltonian $H_1$. Thus, two control channels $H_1$ and $H_2$ are used in this paper. For simplicity, we set the control law related to $H_1$ at a constant value of $1$ and only design the control law $u^{(2)}(t-\tau)$(denoted as $u^{(2)}$ hereafter). The system model is then described by the following equation:
\begin{equation}
\label{eq11}
\left\{
\begin{aligned}
d\rho_t &=-i[H_0+H_1+H_2u^{(2)},\rho_t ]dt+\Gamma_A \mathcal{D}[A]\rho_t dt\\
&+\sqrt{\eta_A \Gamma_A } \mathcal{H}[A]\rho_t dW_A, \\
dy_t &=dW_A+\sqrt{\eta_A \Gamma_A } {\rm Tr}\left[(A+A^\dagger)\rho_t\right]dt. \\
\end{aligned}
\right.
\end{equation}
We first define the following distance function:
\begin{equation}
\label{eq12}
V(\rho)=1-{\rm Tr}^2 (\rho \rho_d),
\end{equation}
and the following sets with a positive real number $\alpha$:
\begin{equation}
\label{eq13}
\left\{
\begin{aligned}
S_{\alpha} &= \left\{ \rho \in S:V(\rho)=\alpha \right\}, \\
S_{>\alpha} &= \left\{ \rho \in S:\alpha<V(\rho)\leq 1 \right\}, \\
S_{\geq \alpha} &= \left\{ \rho \in S:\alpha \leq V(\rho)\leq 1 \right\}, \\
S_{<\alpha} &= \left\{ \rho \in S:0\leq V(\rho)<\alpha \right\}, \\
S_{\leq \alpha} &= \left\{ \rho \in S:0\leq V(\rho)\leq \alpha \right\}.
\end{aligned}
\right.
\end{equation}
The proposed control strategy includes two steps. In the first step, we choose a special state as the intermediate target state that the system will converge to. In the second step, we drive the system state to the given target state.

The state $\frac{1}{n} I_n$ has a special form that can simplify our calculation, and the distances between any states in \eqref{eq10} and this special state are the same. Hence, it is used as the intermediate target state in this paper. For this task, we define a Lyapunov function that has the form
\begin{equation}
\label{eq14}
V_1(\rho)={\rm Tr} \left[\left(\rho-\frac{1}{n}I_n \right)^2 \right].
\end{equation}
When the control law is designed to be constant, we have a deterministic evolution equation from \eqref{eq11} by taking $\bar{\rho_t}=E\{\rho_t\}$ as follows
\begin{equation}
\label{eq14_1}
\dot{\bar{\rho}}_t=-i[H_0+H_1+u^{(2)}H_2, \bar{\rho}_t]+\Gamma_A \mathcal{D}[A]\bar{\rho}_t.
\end{equation}
Calculating the derivative of the Lyapunov function in \eqref{eq14} based on \eqref{eq14_1}, we design the constant control law $u^{(2)}=1$ such that this derivative is non-positive (for details, see \cite{R15}).

The system state might leave from this state $\rho=\frac{1}{n} I_n$ after it converges because of the randomness caused by measurement. To deal with this problem, we define a set $S_{< 1-\gamma}$ containing this intermediate target state, i.e., $0\leq V(\frac{1}{n} I_n )< 1-\gamma,0\leq \gamma < \frac{1}{n^2}$. When the system state converges to $\rho =\frac{1}{n} I_n$ it accordingly enters this set. Thus, we divide the state space into $S_{< 1-\gamma}$ and $S_{\geq 1-\gamma}$, which simultaneously separates the target state with other entangled states; i.e., $V(\rho_d )=0\in S_{< 1-\gamma}$ and $V(\rho_t )=1 \in S_{\geq 1-\gamma}$ for other GHZ entangled states.

The second step is that when the system state converges to $\rho=\frac{1}{n} I_n$ or remains in the set $S_{< 1-\gamma}$, we design a control law that will drive the system to the target state. Therefore we use \eqref{eq12} as the Lyapunov function and calculate its infinitesimal generator:
\begin{equation}
\label{eq15}
\begin{aligned}
\mathcal{L}V(\rho_t ) &=2{\rm Tr}(\rho_d \rho_t ){\rm Tr}\left(i \left[H_0+H_1+H_2 u^{(2)},\rho_t \right] \rho_d \right)\\
&-\eta_A \Gamma_A {\rm Tr}^2 (\mathcal{H}[A]\rho_t \rho_d ).
\end{aligned}
\end{equation}
In order to design a control law that ensures $\mathcal{L}V(\rho_t )\leq 0$, we give a first condition that the control Hamiltonian $H_1$ should satisfy; i.e., $[H_0+H_1, \rho_d]=0$. Therefore the term $2{\rm Tr}(\rho_d \rho_t ){\rm Tr}\left(i [H_0+H_1,\rho_t] \rho_d \right)$ disappears for any $\rho_t$ and \eqref{eq15} is simplified as:
\begin{equation}
\label{eq16}
\begin{aligned}
\mathcal{L}V(\rho_t ) &=2{\rm Tr}(\rho_d \rho_t ){\rm Tr}\left(i \left[H_2 ,\rho_t \right] \rho_d \right)u^{(2)}\\
&-\eta_A \Gamma_A {\rm Tr}^2 (\mathcal{H}[A]\rho_t \rho_d ).
\end{aligned}
\end{equation}
Since the second term $-\eta_A \Gamma_A {\rm Tr}^2 (\mathcal{H}[A]\rho_t \rho_d )$ in \eqref{eq16} remains non-positive during the whole evolution, we choose $u^{(2)}=0$ to ensure $\mathcal{L}V(\rho_t )\leq 0$.

In consideration of the randomness caused by measurement, we design a switching control law as follows:
\begin{enumerate}
\item[1)] If $\rho_{t-\tau} \in S_{\geq {1-\frac{\gamma}{2}}}$ or $\rho_{t-\tau}$ enters $S_{\geq{1-\gamma}}\cap S_{<1-\frac{\gamma}{2}}$ from $S_{\geq{1-\frac{\gamma}{2}}}$, we use the control law $u^{(2)}=1$;
\item[2)] If $\rho_{t-\tau} \in S_{< {1-\gamma}}$ or $\rho_{t-\tau}$ enters $S_{\geq{1-\gamma}}\cap S_{<1-\frac{\gamma}{2}}$ from $S_{<{1-\gamma}}$, we use the control law $u^{(2)}=0$;
\item[3)] If the initial state is in $S_{\geq{1-\gamma}}\cap S_{<1-\frac{\gamma}{2}}$, $u^{(2)}=1$.
\end{enumerate}

The current state of the system usually cannot be obtained immediately due to the delay time in calculating the filter state. Hence, the delayed system state $\rho_{t-\tau}$ is used to switch the control law. The Lyapunov method gives a way to design control laws such that the derivative of the Lyapunov function remains non-positive, while the stability of the whole system should be proved by consideration of proper control Hamiltonians $H_1$ and $H_2$. In the next subsection, we construct the control Hamiltonians by analysing the stability in each set separately.

\subsection{Construction of the control Hamiltonians}
\label{subsec3.1}
When the system state converges to $\rho=\frac{1}{n}I_n$ or remains in $S_{<1-\gamma}$, the infinitesimal generator of the Lyapunov function \eqref{eq12} with the control law $u^{(2)}=0$  satisfies:
\begin{equation}
\label{eq27}
\mathcal{L}V(\rho_t) = - \eta_A \Gamma_A{\rm Tr}^2 (\mathcal{H}[A] \rho_t \rho_d )\leq 0.
\end{equation}
The stochastic LaSalle invariant principle \cite{R35} implies that the system state will converge to the invariant set contained in $\{\rho: \mathcal{L}V(\rho_t)=0\}$. Since the measurement efficiency and strength are normally not zero, we obtain
\begin{equation}
\label{eq33}
{\rm Tr} (\mathcal{H}[A] \rho_t \rho_d )=0.
\end{equation}
That is 
\begin{equation}
\label{eq35}
{\rm Tr}(\mathcal{H}[A] \rho_t \rho_d )=2(\lambda_d-{\rm Tr}(A\rho_t )) {\rm Tr}(\rho_d\rho_t )=0.
\end{equation}
Equation \eqref{eq35} means $\lambda_d={\rm Tr}(A\rho_t )$ or ${\rm Tr}(\rho_d \rho_t )=0$. Since $\rho_t$ is in $S_{<1-\gamma}$, the former holds, which means ${\rm Tr}(A\rho_t )$ is a constant in this process. We calculate $d{\rm Tr}(A\rho_t )$ as:
\begin{equation}
\label{eq36}
\begin{aligned}
d{\rm Tr}(A\rho_t )& =-i{\rm Tr}(A[H_0+H_1,\rho_t ])dt\\
& +2 \sqrt{\eta_A \Gamma_A} ({\rm Tr}(A^2 \rho_t )-{\rm Tr}^2 (A\rho_t ))dW_A\\
&\equiv 0.
\end{aligned}
\end{equation}
That means $-i{\rm Tr}(A[H_0+H_1,\rho_t ])=0$ and
\begin{equation}
\label{eq37}
{\rm Tr}(A^2 \rho_t )-{\rm Tr}^2 (A\rho_t )=0.
\end{equation} 
Equation \eqref{eq37} holds if and only if the system state converges to the $\lambda_d$-related eigenspace of $A$ based on Theorem 1 in \cite{R15}. In this eigenspace, the density operator has a special form, where all the elements in the density operator except $\rho_{11}$, $\rho_{nn}$, $\rho_{1n}$ and $\rho_{n1}$ are zero. Moreover, in this eigenspace we have $\mathcal{D}[A] \rho_t=0$ and $\mathcal{H}[A] \rho_t=0$. Hence, the evolution of the system state becomes:
\begin{equation}
\label{eq38}
d\rho_t=-i[H_0+H_1,\rho_t ]dt.
\end{equation}
We then construct the control Hamiltonian $H_1$ by guaranteeing that the target state $\rho_d$ is the only equilibrium of \eqref{eq38} (for details, see Appendix $A$ in \cite{R15}). After constructing $H_1$, we construct $H_2$ by a similar way in the following.

In the first control step, we use the control law $u^{(2)}=1$ to stabilise the intermediate target state $\rho=\frac{1}{n}I_n$. In this case, the evolution equation of the average states $\bar{\rho}_t$ becomes
\begin{equation}
\label{eq39}
d \bar{\rho}_t=-i[H_0+H_1+H_2,\rho_t ]dt+\Gamma_A \mathcal{D}[A] \bar{\rho}_t dt.
\end{equation}
The time derivative of the Lyapunov function \eqref{eq14} is $\dot{V}_1(\bar{\rho}_t)=-\Gamma_A\left\| \left[A, \bar{\rho}_t \right] \right\| _F^2 \leq 0$. According to the LaSalle invariant principle \cite{R31}, the average system state finally converges to the set $M$ contained in $\left\{\rho:\dot{V}_1 (\bar{\rho})=0\right\}=\left\{\bar{\rho}:[A,\bar{\rho}]=0\right\}$. In the set $M$, $\mathcal{D}[A] \bar{\rho}_t^M=0$, and \eqref{eq39} can be simplified as:
\begin{equation}
\label{eq40}
d \bar{\rho}_t^M=-i[H_0+H_1+H_2,\bar{\rho}_t^M ]dt.
\end{equation}
Similarly, we can construct the control Hamiltonian $H_2$ by guaranteeing that the intermediate target $\bar{\rho}=\frac{1}{n} I_n$ is the only equilibrium point of \eqref{eq40} (for details, see Appendix $B$ in \cite{R15}).

\subsection{Stability analysis}

The following theorem illustrates the stability of the stochastic system with the proposed time delay control strategy.

\begin{theorem}
\label{theorem2}
For the stochastic quantum system \eqref{eq11} with a constant delay time $\tau$, suppose the target state is as in \eqref{eq10}. The measured observable $A$ is given in \eqref{eq9}, and the control Hamiltonians $H_1$ and $H_2$ are constructed such that \eqref{eq38} and \eqref{eq40} have the only equilibrium points $\rho=\rho_d$ and $\rho=\frac{1}{n} I_n$, respectively. Then the following control law guarantees that the system state converges to the target state $\rho_d$ with probability $1$:
\begin{enumerate}
\item[1.] $u^{(2)}=1$ if $\rho_{t-\tau}\in S_{\geq 1-\frac{\gamma}{2}}$, or $\rho_{t-\tau}$ enters $S_{<1-\frac{\gamma}{2}}\cap S_{\geq 1-\gamma}$ from $S_{\geq 1-\frac{\gamma}{2}}$;
\item[2.] $u^{(2)}=0$ if $\rho_{t-\tau}\in S_{<1-\gamma}$, or $\rho_{t-\tau}$ enters $S_{<1-\frac{\gamma}{2}}\cap S_{\geq 1-\gamma}$ from $S_{<1-\gamma}$;
\item[3.] $u^{(2)}=1$ if the initial state is in $S_{<1-\frac{\gamma}{2}}\cap S_{\geq 1-\gamma}$.
\end{enumerate}
\end{theorem}
\begin{proof}
The proof includes three steps.
\begin{enumerate}[fullwidth,itemindent=2em]
\item[Step 1.] When $\rho_{t-\tau}\in S_{\geq 1-\frac{\gamma}{2}}$ or $\rho_{t-\tau}$ enters $S_{\geq 1-\gamma}\cap S_{<1-\frac{\gamma}{2}}$ from $S_{\geq 1-\frac{\gamma}{2}}$, we use $u^{(2)}=1$ such that the system state converges to the intermediate target state $\rho=\frac{1}{n}I_n$. The infinitesimal generator of the Lyapunov function \eqref{eq14} is non-positive such that the system state finally converges to the set $M$. In this set, the evolution equation of system state is simplified as \eqref{eq40}. The control Hamiltonians $H_1$ and $H_2$ ensure that \eqref{eq40} has unique equilibrium point $\rho=\frac{1}{n}I_n$, which means the system state converges to $\rho=\frac{1}{n}I_n$ in probability. We first present the following conclusion (the detailed proof is presented in Appendix).

\begin{proposition}
\label{proposition1}
For the stochastic quantum system \eqref{eq7}, assume that the continuous and differentiable function $V(\rho_t)$ is a given Lyapunov function, and satisfies $0\leq V(\rho_t )\leq 1$. Also, assume the target state $\rho_d$ satisfies $V(\rho_d )=0$. If the system state converges to the target state $\rho_d$ in probability, then it must converge to $\rho_d$ with probability $1$, vice versa.
\end{proposition}

According to Proposition \ref{proposition1}, the system state converges to $\rho=\frac{1}{n}I_n$ with probability $1$; i.e., the system state enters $S_{<1-\gamma}$ with probability $1$.

\item[Step 2.] When $\rho_{t-\tau}\in S_{<1-\gamma}$ or $\rho_{t-\tau}$ enters $S_{\geq 1-\gamma}\cap S_{<1-\frac{\gamma}{2}}$ from $S_{<1-\gamma}$, we apply the control law $u^{(2)}=0$. With the consideration of randomness, the system state that has converged to $\rho=\frac{1}{n}I_n$ or entered $S_{<1-\gamma}$ might leave from this intermediate state or this set. We need to prove that states under the control law $u^{(2)}=0$ finally remain in set $S_{<1-\frac{\gamma}{2}}$ with probability $1$.

We denote the solution of the system \eqref{eq11} with control $u^{(2)}$ and initial state $\rho=\rho_0$ as $\Phi(\rho_0,u^{(2)} )$. As we consider the fact that states are in $S_{<1-\gamma}$, $V(\rho_0 )<1-\gamma$ holds. Based on \eqref{eq27}, we have \cite{R35}:
\begin{equation}
\label{eq43}
\begin{aligned}
& P\left\{{\rm sup}_{t\geq \tau}V(\Phi(\rho_0,u^{(2)}))\geq 1-\frac{\gamma}{2}
\right\}\\
& \leq \frac{1-\gamma}{1-\frac{\gamma}{2}} \triangleq 1-p.
\end{aligned}
\end{equation}
In \eqref{eq43}, $0<\gamma<1$ and $0<1-\frac{\gamma}{2}<1-\gamma$. Hence, we have $1-p<1,0<p<1$, and \eqref{eq43} means 
$$P\left\{{\rm sup}_{t\geq -\tau}V(\Phi(\rho_0,u^{(2)}))<1-\frac{\gamma}{2}\right\}\geq p ;$$
i.e., the system trajectory remains in $S_{<1-\frac{\gamma}{2}}$ with the probability that is greater than or equals to $p$. According to the Borel-Centelli lemma \cite{R20}, the switching times between $S_{<1-\frac{\gamma}{2}}$ and $S_{\geq 1-\frac{\gamma}{2}}$ are finite. Hence, the system state remains in $S_{<1-\frac{\gamma}{2}}$ with probability $1$. %It should be noted that the value of $\gamma$ will be smaller as the dimension goes higher, which causes that the set $S_{<1-\gamma}$ contains more states. Hence, the probability $1-p$ in \eqref{eq43} becomes bigger, which results in more switching times. That means the convergent time becomes longer as the system dimension goes higher. It is the intermediate target that poses this strict condition on $\gamma$, it is possible to design an improved control strategy in the first step to further speed up the convergence.

\item[Step 3.] The final step is to prove the system states that remain in $S_{<1-\frac{\gamma}{2}}$ will converge to $\rho_d$ with probability $1$ under the influence of the control law $u^{(2)}=0$.

When we apply the control $u^{(2)}=0$, the infinitesimal generator of the Lyapunov function \eqref{eq12} satisfies \eqref{eq27}. Then, the system state will converge to $\left\{\rho_t: \mathcal{L}V (\rho_t )=0 \right\}$ with probability $1$ according to \eqref{eq36} and \eqref{eq37}. In this set, the evolution equation becomes \eqref{eq38}, and $\rho_d$ is the only stable equilibrium point of \eqref{eq38}. The stochastic LaSalle invariant principle ensures that the target state $\rho_d$ is also attractive; i.e., the system state converges to $\rho_d$ with probability $1$. $\blacksquare$
\end{enumerate}
\end{proof}

\section{Switching Lyapunov control with time delay}
\label{sec4}
In this section, we propose another control strategy switching between the constant value $1$ and the Lyapunov feedback control law, which has the potential to speed up convergence to the target state compared with the bang-bang-like case. 

When the system state is in $S_{\geq 1-\gamma}$, the control law has the constant value $1$ as in the bang-bang-like case. Hence, we only consider the case where the state is in the $S_{< 1-\gamma}$ in the following. When the system state converges to $\rho=\frac{1}{n} I_n$ or remains in the set $S_{< 1-\gamma}$, the derivative of the Lyapunov function $V(\rho)$ satisfies \eqref{eq16}. For delay-free cases, the control law is designed as $u^{(2)}=-k{\rm Tr}(i[H_2,\rho_t]\rho_d)$ (\cite{R12}, \cite{R15}) to ensure that the first term $2{\rm Tr}(\rho_d \rho_t ){\rm Tr}\left(i \left[H_2 ,\rho_t \right] \rho_d \right)u^{(2)}$ remains non-positive, where $k>0$ is the control gain and is limited by the bound of the control law $u^{(2)}$. However, usually we need some time to compute the estimate state based on the filter equation, especially for high-dimensional systems. This implies that we usually cannot obtain the system state $\rho_t$ at current time $t$ as the feedback information to calculate the control signal. Instead, the system state $\rho_{t-\tau}$ is used to design the control law; i.e., we design a control law as:
\begin{equation}
\label{eq44}
u^{(2)}=-k{\rm Tr}(i[H_2,\rho_{t-\tau}]\rho_d).
\end{equation}
Similar to the bang-bang-like case, we design a control law with switching between the constant value $1$ and a delayed Lyapunov feedback control law as follows:
\begin{enumerate}
\item[1.] $u^{(2)}=1$ if $\rho_{t-\tau}\in S_{\geq 1-\frac{\gamma}{2}}$, or $\rho_{t-\tau}$ enters $S_{<1-\frac{\gamma}{2}}\cap S_{\geq 1-\gamma}$ from $S_{\geq 1-\frac{\gamma}{2}}$;
\item[2.] $u^{(2)}=-k{\rm Tr}(i[H_2,\rho_{t-\tau}]\rho_d)$ if $\rho_{t-\tau}\in S_{<1-\gamma}$, or $\rho_{t-\tau}$ enters $S_{<1-\frac{\gamma}{2}}\cap S_{\geq 1-\gamma}$ from $S_{<1-\gamma}$;
\item[3.] $u^{(2)}=1$ if the initial state is in $S_{<1-\frac{\gamma}{2}}\cap S_{\geq 1-\gamma}$.
\end{enumerate}
The stability proof with the control law $u^{(2)}=1$ is the same as that in the bang-bang-like case. Normally, the delay time $\tau$ in the control law $u^{(2)}=-k{\rm Tr}(i[H_2,\rho_{t-\tau}]\rho_d)$ might cause the first term in \eqref{eq16} to be non-negative, which cannot always ensure the decrease of the Lyapunov function $V(\rho)$. Hence, stability with the time-delay control law $u^{(2)}=-k{\rm Tr}(i[H_2,\rho_{t-\tau}]\rho_d)$ needs to be proved. A delay-dependent criterion for differential stochastic equation has been proposed in \cite{R36} , which is described as follows.

\begin{lemma}\cite{R36}
\label{lemma1}
Let $f(\cdot,\cdot): \mathbb{R}^n \times \mathbb{R}^n \rightarrow \mathbb{R}^n$, $g(\cdot):\mathbb{R}^n \rightarrow \mathbb{R}^n$, be polynomials and $\mathsf{C}$ a bounded semi-algebraic set in $\mathbb{R}^n$ such that for any initial condition $\tilde{x}_i \in C_{\mathsf{C}}^{\tau}$ the solution to the delay differential stochastic equation
\begin{equation}
\label{eq45}
\begin{aligned}
dx_t&=f(x_t,x_{t-\tau})dt+g(x_t)dw_t\\
x_{\theta}&=\tilde{x}_i(\theta)\in \mathsf{C}, \theta \in [-\tau,0],
\end{aligned}
\end{equation}
does not exit $\mathsf{C}$ almost surely. Suppose there exist a polynomial $V_*(\cdot)$ which is positive in $\mathbb{C}$, $n$-variable polynomials $V_i(i=0,1), S\in \mathbb{R}^{2n\times n}$, and positive-definite matrices $R,T\in \mathbb{R}^{n\times b}$ such that $\Upsilon$ defined below is negative in $\mathsf{C}\times \mathsf{C}\times \mathbb{R}^{2n}$
\begin{equation}
\label{eq46}
\begin{aligned}
F(x,x_d)&:=\left(\frac{\partial V_0(x)}{\partial x} \right)^T f(x,x_d)\\
&+\frac{1}{2}g(x)^T \frac{\partial}{\partial x}\left(\frac{\partial V_0(x)}{\partial x} \right)^T g(x)\\
&+V_1(x)-V_1(x_d)+V_*(x)+\tau \|g(x)\|_T^2\\
&+2\begin{bmatrix}
x^T & x_d^T
\end{bmatrix}S(x-x_d)+\tau\|f(x,x_d) \|_R^2,\\
\Upsilon(x,x_d,y)&:=F(x,x_d)+\\
&\begin{bmatrix}
x\\
x_d\\ \hline
y
\end{bmatrix}^T \times \left[\begin{array}{c|c c}
0 & S & \tau S\\ \hline
S^T & -T & 0\\
\tau S^T & 0 & -\tau R 
\end{array}\right] \times \begin{bmatrix}
x\\
x_d\\ \hline
y
\end{bmatrix},
\end{aligned}
\end{equation}
Then, $V_*(x_t)$ converges to $0$ almost surely for any initial condition $\tilde{x}_i \in C_{\mathsf{C}}^{\tau}$.
\end{lemma} 
Using Lemma \ref{lemma1}, we can define the positive polynomial $V_*(\cdot)$ as a distance function between the system states and the target state. Then this lemma states a sufficient condition to ensure that the system state $x_t$ converges to the target state if a semi-algebraic problem is feasible. Semi-algebraic problems are in general NP-hard and the possible solutions are not unique. Hence, it seems that it is not easy to give an analytical solution to this problem. The MATLAB SOS-TOOLS has been used to check the feasibility for a special spin system in \cite{R36}. Inspired by the proof of Lemma \ref{lemma1} in \cite{R36} and the proof of Proposition 2 in \cite{R19}, we provide a theoretical analysis of the stability of a special two-qubit system as an example in the following sub-section.

\subsection{A two-qubit example}
\label{subsec4.1} 
We write the density matrix of a two-qubit system as 
\begin{equation}
\label{eq47}
\rho=\begin{bmatrix}
\nu_1 & \lambda_1-i\mu_1 & \lambda_2-i\mu_2 & \lambda_3-i\mu_3\\
\lambda_1+i\mu_1 & \nu_2 & \lambda_4-i\mu_4 & \lambda_5-i\mu_5\\
\lambda_2+i\mu_2 & \lambda_4+i\mu_4 & \nu_3 & \lambda_6-i\mu_6\\
\lambda_3+i\mu_3 & \lambda_5+i\mu_5 & \lambda_6+i\mu_6 & 1-\nu_1-\nu_2-\nu_3
\end{bmatrix},
\end{equation} with all $\nu_i, \lambda_i$ and $\mu_i$ being real and scalar.

Suppose the free Hamiltonian is $H_0={\rm diag}[1,-1,-1,1]$, and the target state is the following Bell state
\begin{equation}
\label{eq47.1}
\rho_d=\frac{1}{2}\left[\begin{smallmatrix}
1 & 0 &  0 & 1\\
0 & 0 & 0 & 0\\
0 & 0 & 0 & 0 \\
1 & 0 & 0 &1
\end{smallmatrix}\right].
\end{equation}
According to \eqref{eq9} and the conditions that control Hamiltonians should satisfy, we choose the observable operator as $A=\sigma_z \otimes \sigma_z={\rm diag}[1,-1,-1,1]$, and control Hamiltonians as
$$H_1=I\otimes \sigma_x-\sigma_x \otimes I=\left[\begin{smallmatrix}
0& 1& -1& 0\\
1& 0& 0& -1\\
-1& 0& 0& 1\\
0& -1& 1& 0\\
\end{smallmatrix}\right],$$
$H_2=\sigma_z \otimes I={\rm diag}[1,1,-1,-1]$. Substituting these two control Hamiltonians into \eqref{eq38} and \eqref{eq40}, it can be verified that $\rho=\rho_d$ and $\rho=\frac{1}{n} I_n$ are the only equilibrium points of \eqref{eq38} and \eqref{eq40}, respectively.

Substituting the control Hamiltonians $H_1$ and $H_2$ into \eqref{eq11} and doing straightforward calculations, we may write the dynamical equation for $\mu_3$ as:
\begin{equation}
\label{eq48}
\begin{aligned}
d\mu_3(t)&=(\lambda_1(t)-\lambda_2(t)+\lambda_5(t)-\lambda_6(t)+2\lambda_3(t) u^{(2)})dt\\
&+4\mu_3(t)(\nu_2(t)+\nu_3(t))dW_t.
\end{aligned}
\end{equation}
The control law is designed as $u^{(2)}=-k{\rm Tr}(i[H_2,\rho_{t-\tau}]\rho_d)=-2k\mu_3(t-\tau)$. With this control law, the stability of this time-delay stochastic system is then described by the following theorem.
\begin{theorem}
\label{theorem3}
For a given delay time $\tau\geq 0$, suppose that there exists a feedback gain $k>0$ such that the following inequality has solution for $r,q,\epsilon>0, S\in \mathbb{R}^{3\times 1}$:
\begin{equation}
\label{eq1}
\left[\begin{smallmatrix}
M+S\tilde{S} & S & \tau S\\
S^T & -\epsilon & 0\\
\tau S^T & 0 & -\tau r
\end{smallmatrix}\right]<0,
\end{equation}
$$
M=\left[\begin{smallmatrix}
q+\tau \epsilon & -2k & 0\\
-2k & -q & 0\\
0 & 0 & \tau r
\end{smallmatrix}\right],  \quad \tilde{S}=\begin{bmatrix} 2 & -2 & 0 \end{bmatrix}.$$ 
Then the following switching control law globally stabilises the target state $\rho_d$:
\begin{enumerate}
\item[1.] $u^{(2)}=1$ if $\rho_{t-\tau}\in S_{\geq 1-\frac{\gamma}{2}}$, or $\rho_{t-\tau}$ enters $S_{<1-\frac{\gamma}{2}}\cap S_{\geq 1-\gamma}$ from $S_{\geq 1-\frac{\gamma}{2}}$;
\item[2.] $u^{(2)}=-k{\rm Tr}(i[H_2,\rho_{t-\tau}]\rho_d)$ if $\rho_{t-\tau}\in S_{<1-\gamma}$, or $\rho_{t-\tau}$ enters $S_{<1-\frac{\gamma}{2}}\cap S_{\geq 1-\gamma}$ from $S_{<1-\gamma}$;
\item[3.] $u^{(2)}=1$ if the initial state is in $S_{<1-\frac{\gamma}{2}}\cap S_{\geq 1-\gamma}$.
\end{enumerate}
\end{theorem}

\begin{proof}
Compared with the stability proof in the bang-bang-like case, the only difference here is that we cannot ensure that \eqref{eq16} is always non-positive due to the delay-dependent term $-2k{\rm Tr}(\rho_t \rho_d){{\rm Tr}(i[H_2,\rho_t]\rho_d)\rm Tr}(i[H_2,\rho_{t-\tau}]\rho_d)$. Hence, we need to find conditions such that \eqref{eq16} remains non-positive under this delay control law in the following proof. Once these conditions are satisfied, the proof of bang-bang-like case carries over directly.  

We use the notation: $\tilde{(\cdot)}_t(\theta)=(\cdot)(t+\theta), \theta \in [-2\tau,0]$ hereafter, and denote $\lambda_1(t)-\lambda_2(t)+\lambda_5(t)-\lambda_6(t)+2\lambda_3(t) u^{(2)}=f(t)$ and $\mu_3(t)(4\nu_2(t)+4\nu_3(t))=g(t)$. For simplicity, we also write $\mu_3(t)$ and $\mu_3(t-\tau)$ as $\tilde{\mu}_t(0)$ and $\tilde{\mu}_t(-\tau)$, respectively. Also, suppose the initial conditions are chosen without loss of generality as $f_t=0$, and $\mu_t=-\frac{1}{k}$ for $t\in [-\tau, 0)$.

Since $1-{\rm Tr}(\rho_t \rho_d)=\frac{1}{2}\tilde{\nu}_2(0)+\frac{1}{2}\tilde{\nu}_3(0)-\tilde{\lambda}_3(0)+\frac{1}{2}$, we construct a Lyapunov function with positive $q,r$ and $\epsilon$ as follows:
\begin{equation}
\label{eq49}
\begin{aligned}
V_2&=\frac{1}{2}\tilde{\nu}_2(0)+\frac{1}{2}\tilde{\nu}_3(0)-\tilde{\lambda}_3(0)+\frac{1}{2}\\
& +q\int_{-\tau}^0 |\tilde{\mu}(0)|^2 ds+\int_{-\tau}^0 \int_s^0 \{r|\tilde{f}(v)|^2+16\epsilon|\tilde{\mu}(v)|^2 \}dvds.
\end{aligned}
\end{equation}
By defining this Lyapunov function, we are setting $V_0=\frac{1}{2}\tilde{\nu}_2(0)+\frac{1}{2}\tilde{\nu}_3(0)-\tilde{\lambda}_3(0)+\frac{1}{2}, V_1=|\tilde{\mu}_t(0)|^2, \|f(\tilde{x}(\theta),\tilde{x}(-\tau+\theta))\|_R^2=|\tilde{f}(\theta)|^2$ and $\|g(\tilde{x}(\theta))\|_T^2=|\tilde{\mu}_t(\theta)|^2$ in the proof of Theorem 1 in \cite{R36}. It should be noted that Lemma \ref{lemma1} only gives sufficient conditions, where $R$ and $T$ in this two-qubit example can be chosen as semi-positive as long as one can ensure the Lyapunov function in \eqref{eq49} is non-negative. We may calculate the infinitesimal generator of $V_2$ as
\begin{equation}
\label{eq50}
\begin{aligned}
\mathcal{L}V_2(t,\tilde{\rho}) &=-4k\tilde{\mu}(0)\tilde{\mu}(-\tau)+q\{\tilde{\mu}(0)^2-\tilde{\mu}(-\tau)^2\}\\
& +\tau \left[ r|\tilde{f}_t(0)|^2 +16\cdot \epsilon |\tilde{\mu}(0)|^2\right]\\
& -\int_{-\tau}^0\{r|\tilde{f}(s)|^2+16\cdot \epsilon |\tilde{\mu}(s)|^2 \}ds.
\end{aligned}
\end{equation} 
With $e=\begin{bmatrix}
\tilde{\mu}(0)& \tilde{\mu}(-\tau) & \tilde{f}(0)
\end{bmatrix}^T
$ and $\tilde{\xi}(s)=\begin{bmatrix}
e^T & \tilde{f}(s)
\end{bmatrix}^T$, we have the following inequalities \cite{R36}
\begin{equation}
\label{eq51}
0\leq \tau e^TXe-\int_{-\tau}^0 e^TX e ds, \forall X\geq 0,
\end{equation}
\begin{equation}
\label{eq52}
\begin{aligned}
0&=(2-2)e^T S \left\{ \tilde{\mu}(0)-\tilde{\mu}(-\tau)-\int_{-\tau}^0 \tilde{f}(s)ds \right\}\\
& \leq 2e^T S\left( \tilde{\mu}(0)-\tilde{\mu}(-\tau) \right)-\int_{-\tau}^0 2e^T S \tilde{f}(s)ds+e^TS\epsilon^{-1}S^Te\\
&+\epsilon \left| \tilde{\mu}(0)-\tilde{\mu}(-\tau)-\int_{-\tau}^0 \tilde{f}(s)ds \right|^2.
\end{aligned}
\end{equation}
Hence, we obtain
\begin{equation}
\label{eq53}
\begin{aligned}
&\mathcal{L}V_2(t,\tilde{\rho})+0+0+16\cdot \int_{-\tau}^0 \epsilon |\tilde{\mu}(s)|^2ds -\\
&\epsilon \left| \tilde{\mu}(0)-\tilde{\mu}(-\tau)-\int_{-\tau}^0 \tilde{f}(s)ds \right|^2\\
&\leq e^T\left\{ M+\tau X+S\tilde{S}+S\epsilon^{-1}S^T \right\}e-\int_{-\tau}^0 \tilde{\xi}(s)^T \begin{bmatrix}
X & S\\
S^T & r
\end{bmatrix} \tilde{\xi}(s)ds,
\end{aligned}
\end{equation}
where $M=\left[\begin{smallmatrix}
q+\tau \epsilon & -2k & 0\\
-2k & -q & 0\\
0 & 0 & \tau r
\end{smallmatrix}\right]$.

For a given delay time $\tau$, if there exists a feedback gain $k$ such that we can find proper $q,r,\epsilon$ and $S$ satisfying $M+\tau X+S\tilde{S}+S\epsilon^{-1} S^T <0$ and take $X=r^{-1} S S^T$ to ensure $\left[ \begin{smallmatrix}
X & S\\
S^T & r
\end{smallmatrix}
 \right]$ is non-positive. We then have the following inequality with a positive $\beta$ being the smallest singular value of $M$ \cite{R19}
 \begin{equation}
 \label{eq54}
 \begin{aligned}
 \mathcal{L}V_2(t,\tilde{\rho}) & \leq -\beta \tilde{\mu}_t(0)^2 -16\epsilon \int_{-\tau}^0 |\tilde{\mu}_t(s) |^2 ds\\
 &+\epsilon \left| \tilde{\mu}_t(0)-\tilde{\mu}_t(-\tau)-\int_{-\tau}^0 \tilde{f}_t(s)ds \right|^2.
 \end{aligned}
 \end{equation}
We consider the positive term $\epsilon \left| \tilde{\mu}_t(0)-\tilde{\mu}_t(-\tau)-\int_{-\tau}^0 \tilde{f}_t(s)ds \right|^2$ in \eqref{eq54}. When $t\geq \tau$, by the dynamics of $\tilde{\mu}_t(0)=\mu_3(t)$ in \eqref{eq48} and It\^{o} isometry we have 
\begin{equation}
\label{eq55}
\begin{aligned}
& \mathbb{E}\left| \tilde{\mu}_t(0)-\tilde{\mu}_t(-\tau)-\int_{-\tau}^0 \tilde{f}_t(s)ds \right|^2\\
& =\mathbb{E} \left| \int_{-\tau}^0 \tilde{f}_t(s)ds+\int_{-\tau}^0 \tilde{g}(s)dW_s-\int_{-\tau}^0 \tilde{f}_t(s)ds \right|^2\\
 & \leq \mathbb{E}\left\{\int_{-\tau}^0 \tilde{g}_t(s)^2 ds \right\}\leq 16\cdot \mathbb{E}\left\{\int_{-\tau}^0 \tilde{\mu}_t(s)^2 ds \right\}.
\end{aligned}
\end{equation}
Substituting \eqref{eq55} into \eqref{eq54}, we have 
\begin{equation}
\label{eq56}
\mathbb{E}\mathcal{L}V_1 \leq -\beta \mathbb{E} \tilde{\mu}_t(0)^2.
\end{equation}
When $t\in [0,\tau)$,
\begin{equation}
\label{eq57}
\begin{aligned}
& \mathbb{E}\left| \tilde{\mu}_t(0)-\tilde{\mu}_t(-\tau)-\int_{-\tau}^0 \tilde{f}_t(s)ds \right|^2\\
&= \mathbb{E}\left| \mu_t-\mu_0-\int_{-t}^0 \tilde{f}_t(s)ds\right|^2+\left| \mu_0 -\mu_{t-\tau} \right|^2\\
& \leq 16\cdot \mathbb{E} \int_{-t}^0 \tilde{\mu}_t(s)^2ds+\left| \mu_0 -\mu_{t-\tau} \right|^2.
\end{aligned}
\end{equation}
Substituting \eqref{eq57} into \eqref{eq54}, we have
\begin{equation}
\label{eq58}
\begin{aligned}
\mathbb{E} \mathcal{L} V_1 &\leq -\beta \mathbb{E}\tilde{\mu}_t(0)^2 -16\cdot \epsilon_1 \int_{-\tau}^0 |\tilde{\mu}_t(s)|^2ds\\
& +\epsilon_1 \left( 16\cdot \mathbb{E}\int_{-t}^0 \tilde{\mu}_t(s)^2 ds+|\mu_0-\mu_{t-\tau} |^2\right)\\
& \leq \epsilon_1.
\end{aligned}
\end{equation}
From \eqref{eq56} and \eqref{eq58}, we have 
\begin{equation}
\label{eq59}
\mathbb{E} \mathcal{L}V_1 \leq \left\{
\begin{array}{c}
\epsilon_1, \quad t\in [0,\tau)\\
-\beta \mathbb{E}\tilde{\mu}_t(0)^2,  \quad t\geq \tau.
\end{array}
\right.
\end{equation}
According to Lemma 1 in \cite{R19}, $\mu_t$ converges to $0$ with probability $1$. With these system parameters and based on \eqref{eq16}, we have 
\begin{equation}
\label{eq60}
\mathcal{L}V(\rho_t) = -4k \tilde{\mu}_t(0) \tilde{\mu}_t(-\tau)-\eta_A \Gamma_A {\rm }^2\left(\mathcal{H}[A]\rho_t\rho_d \right).
\end{equation}
Since $\mu_t$ converges to $0$ with probability $1$ when $t \rightarrow \infty$, we have $\mathcal{L}V(\rho_t)\leq -\eta_A \Gamma_A {\rm }^2\left(\mathcal{H}[A]\rho_t\rho_d \right) \leq 0$, which means \eqref{eq27} holds. This completes the proof.$\blacksquare$
\end{proof}
\begin{remark}
The stability criterion for this two-qubit system has been given in the form of LMIs, which enables us to systematically search for a feedback gain for a given delay time such that the target state is globally stabilised. For two-qubit quantum systems, we assume that the dephasing noise is described as $\Delta(t)=\beta(t)(\sigma_z\times \sigma_z)$, where $\beta(t)$ represents a stochastic process \cite{dephasing}. We obtain the infinitesimal generator of the Lyapunov function $V(\rho)$ as $\mathcal{L}V(\rho_t)=2{\rm Tr}(\rho_t\rho_d){\rm Tr}(i[H_2,\rho_t]\rho_d)u^{(2)}-\eta_A\Gamma_A {\rm Tr}^2(\mathcal{H}[A]\rho_t \rho_d)$ since ${\rm Tr}(i[\Delta(t),\rho_t]\rho_d)=0$. Hence, the proposed switching Lyapunov control law has good robustness in terms of dephasing noise for two-qubits systems.
\end{remark}
%%%%%%%%%%%%%%%%%%%%%%%%%%%%%%%%%%%%%%%%%%%%%%%%

\section{Numerical examples}
\label{sec5}
Now, we present numerical results on a two-qubit system for two different control strategies. The system parameters of this two-qubit system are the same as in Section \ref{subsec4.1}.

Suppose the delay time is $\tau=0.2$ (in unit with $\bar{h}=1$), and the target state is as in \eqref{eq47.1}. We consider two initial states $\rho_0^{(1)}={\rm diag}[0,1,0,0]\in S_{\geq {1-\frac{\gamma}{2}}}$ and $\rho_0^{(2)}={\rm diag}[1,0,0,0] \in S_{<1-\gamma}$, and set the measurement strength and efficiency as $\eta_A=1$ and $\Gamma_A=1$, respectively. We choose $\gamma=0.06$ and feedback gain as $k=1$. Fig. \ref{Fig.1} and Fig. \ref{Fig.2} show the results of two different control strategies with the same initial state $\rho_0^{(1)}$. We have repeated 30 simulations for each control strategy, and only show the average evolution curves (black and solid lines) and two special samples for clarity.

%%%%%%%%%%%%%%%%%%%%%%%%%%%%
\begin{figure}
\centering
\includegraphics[width=\textwidth]{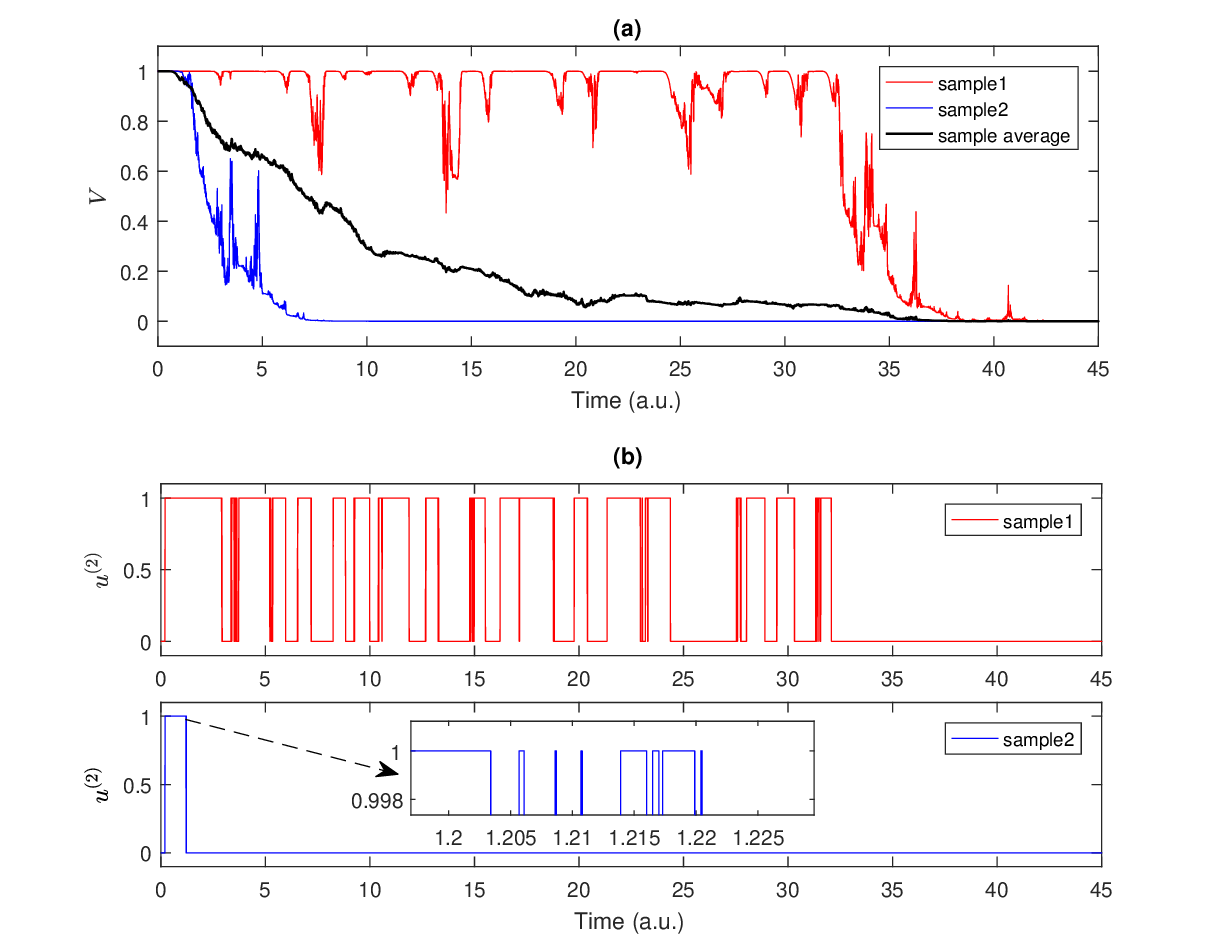}
\caption{(a) The evolution curve of the distance $V=1-{\rm Tr}^2(\rho \rho_d )$ between the system state and the target state with bang-bang-like control for initial state $\rho_0^{(1)}$ and $\tau=0.2$, and (b) the evolution curves of the control field $u^{(2)}$.\label{Fig.1}}
\end{figure}
%%%%%%%%%%%%%%%%%%%%%%%%%
\begin{figure}
\centering
\includegraphics[width=\textwidth]{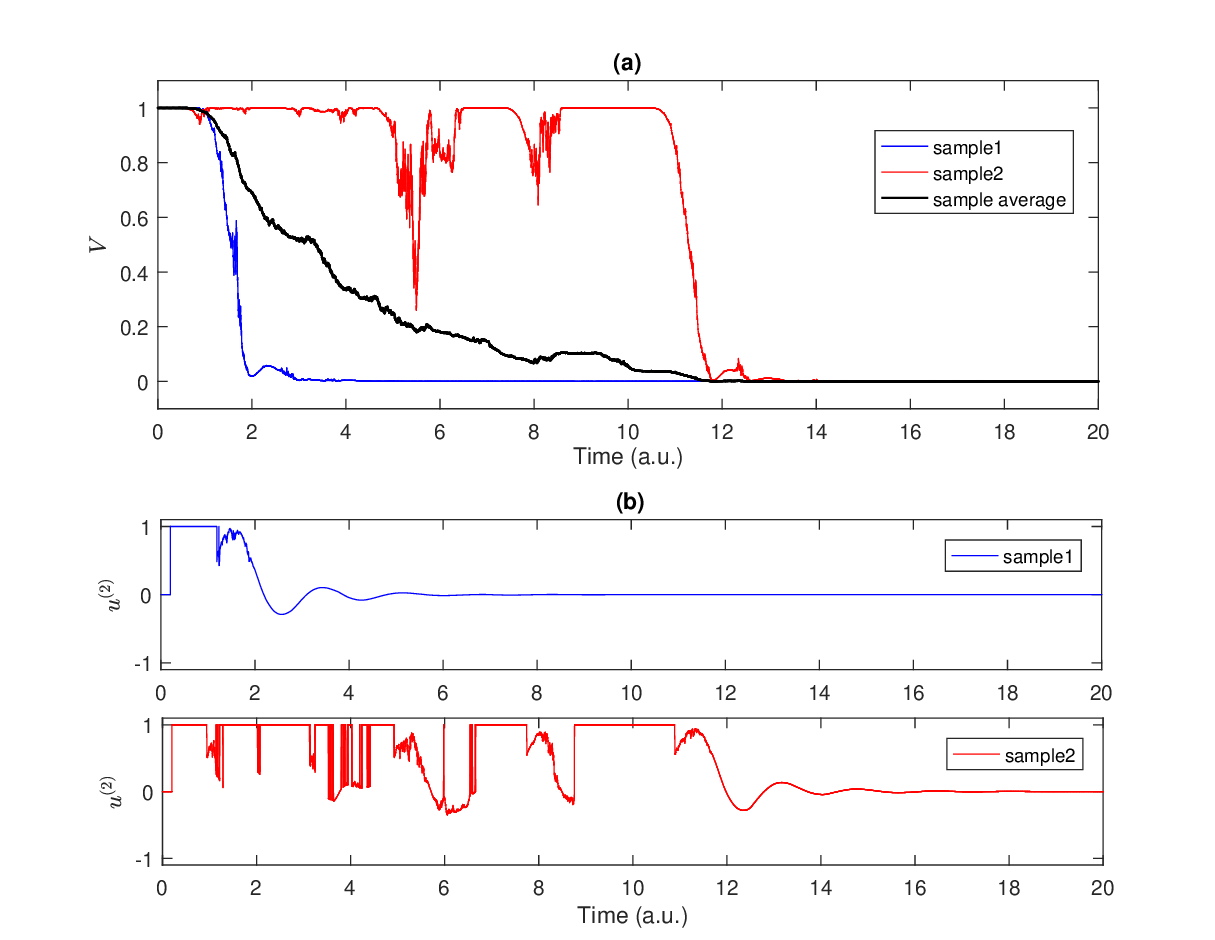}
\caption{(a) The evolution curve of the distance $V=1-{\rm Tr}^2(\rho \rho_d )$ between the system state and the target state with switching Lyapunov control for initial state $\rho_0^{(1)}$ and $\tau=0.2$, and (b) the evolution curves of the control field $u^{(2)}$.\label{Fig.2}}
\end{figure}
From Fig. \ref{Fig.1} and Fig. \ref{Fig.2}, we can see that the distances between the system state and the target state (for every single sample and the average result) converge to zero, which means that the system state converges to the target state. In these two figures, the initial state lies in  $S_{\geq {1-\frac{\gamma}{2}}}$. Hence, both the control laws start from the constant value $u^{(2)}=1$. 

Theoretically, the proposed two control strategies also work for the cases of imperfect measurement; i.e., $\eta<1$, and dephasing noise. The results for $\eta=0.8$ are shown in Fig. \ref{Fig.3}, and the results for dephasing noise $\Delta(t)=\beta(t)(\sigma_z\otimes \sigma_z)$ are presented in Fig. \ref{Fig.4}, where the initial state is taken as $\rho_0^{(2)}$. Similarly, we only show the average results over 30 repetitions and two specific samples for clarity. The results show that the control laws designed using the two control strategies can stabilise the target state effectively even when the measurement efficiency is less than one, or there exists dephasing noise.

\begin{figure}
\centering
\includegraphics[width=\textwidth]{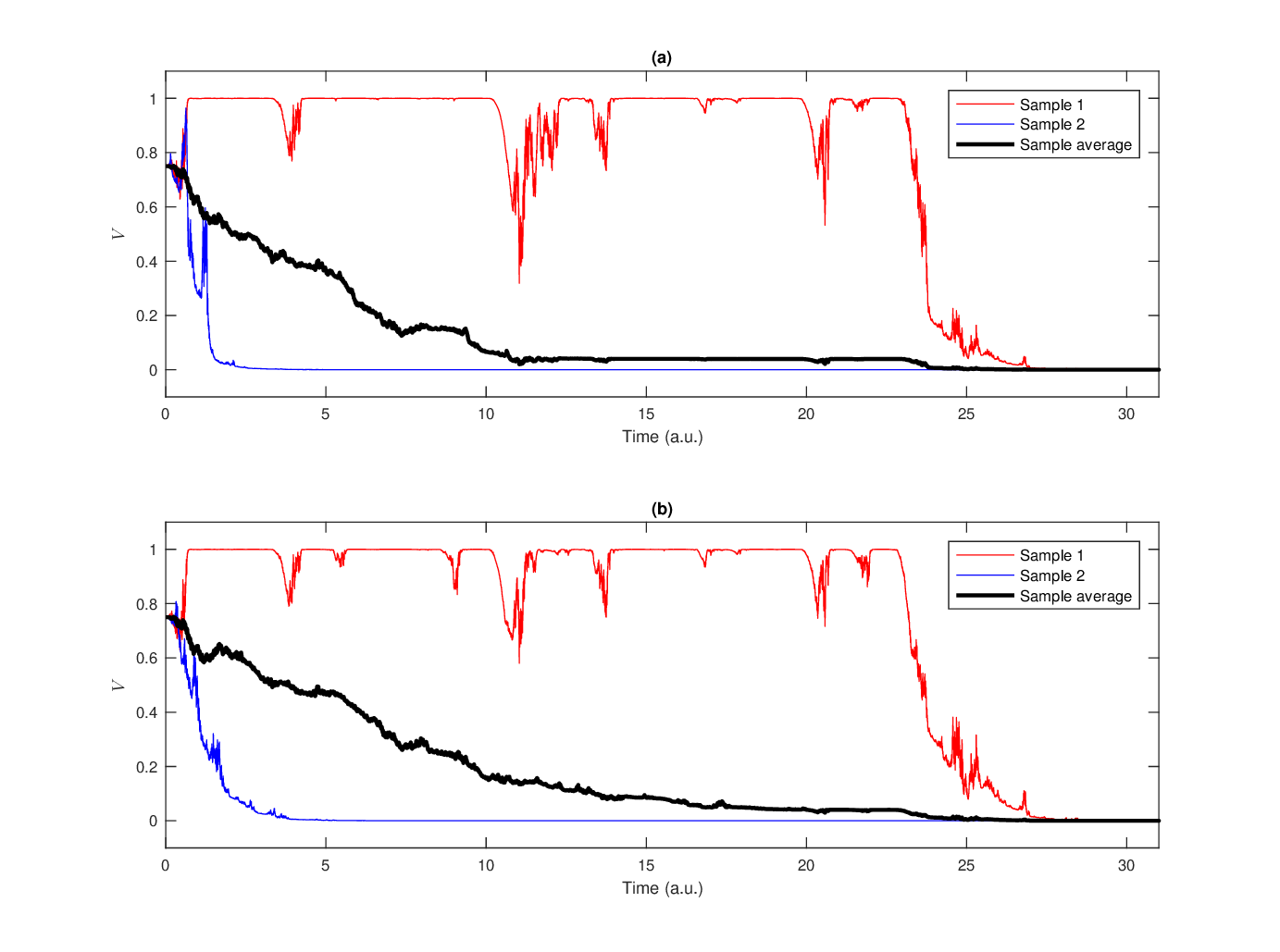}
\caption{The results with imperfect measurement $\eta_A=0.8$; (a)  The evolution curve for the distance $V=1-{\rm Tr}^2(\rho \rho_d )$ with switching Lyapunov control for initial state $\rho_0^{(2)}$ and $\tau=0.2$, (b) The evolution curve for the distance $V=1-{\rm Tr}^2(\rho \rho_d )$ with bang-bang-like control for initial state $\rho_0^{(2)}$ and $\tau=0.2$. \label{Fig.3}}
\end{figure}

\begin{figure}
\centering
\includegraphics[width=\textwidth]{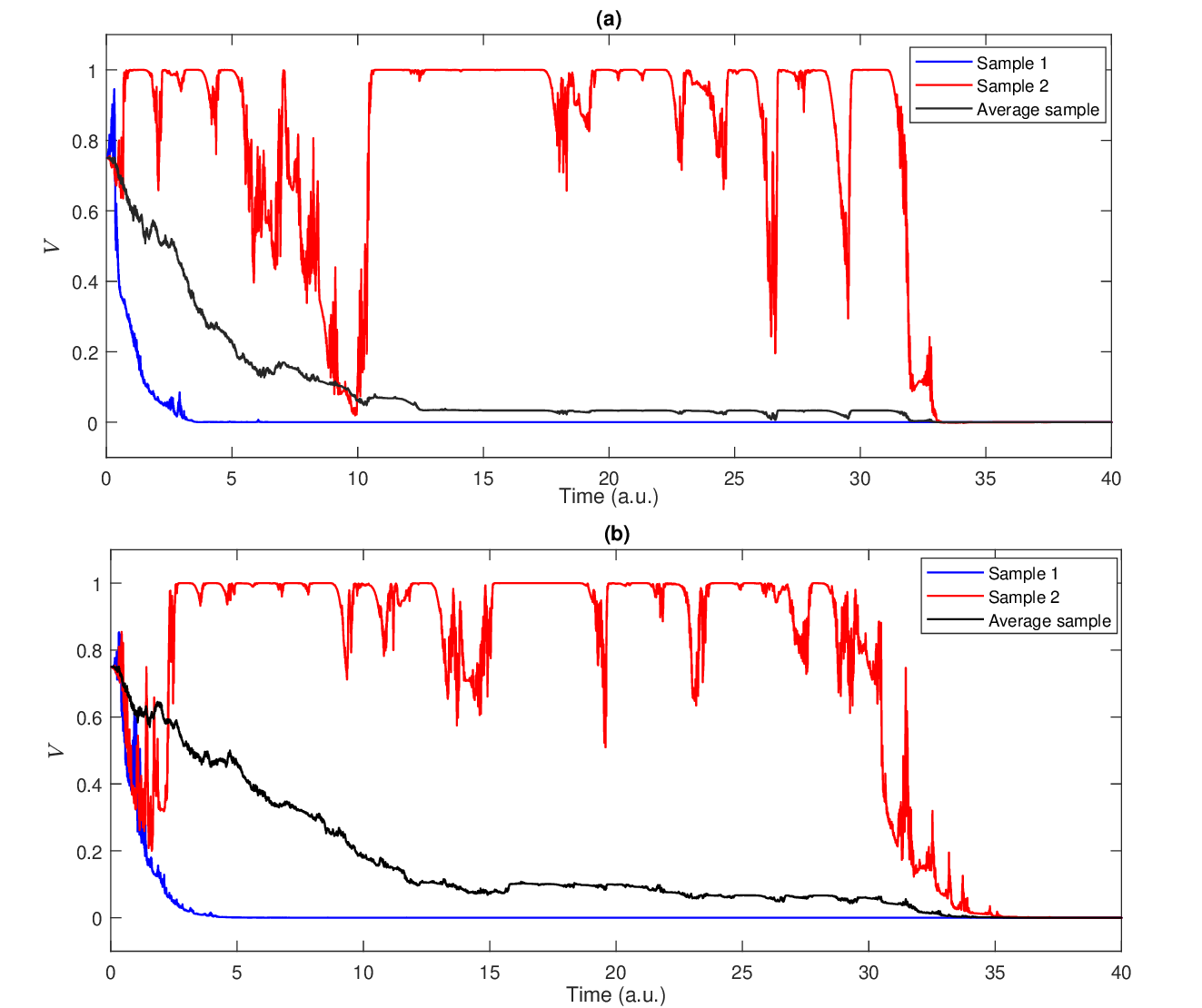}
\caption{The results with dephasing noise; (a)  The evolution curve for the distance $V=1-{\rm Tr}^2(\rho \rho_d )$ with switching Lyapunov control for initial state $\rho_0^{(2)}$ and $\tau=0.2$, (b) The evolution curve for the distance $V=1-{\rm Tr}^2(\rho \rho_d )$ with bang-bang-like control for initial state $\rho_0^{(2)}$ and $\tau=0.2$.} \label{Fig.4}
\end{figure}

In Fig. \ref{Fig.5}, we compare the convergence performance between these two control strategies. The simulation curves in Fig. \ref{Fig.5} are the average result over 30 samples starting from the initial state $\rho_0^{(1)}$. It is clear that the filter-based feedback control $u^{(2)}=-{\rm Tr}(i[H_2,\rho_{t-\tau}]\rho_d)$ can speed up the convergence to the target state and has improved performance. The comparison for the initial state $\rho_0^{(2)}$ has a similar result to Fig. \ref{Fig.5}.

\begin{figure}
\centering
\includegraphics[width=\textwidth]{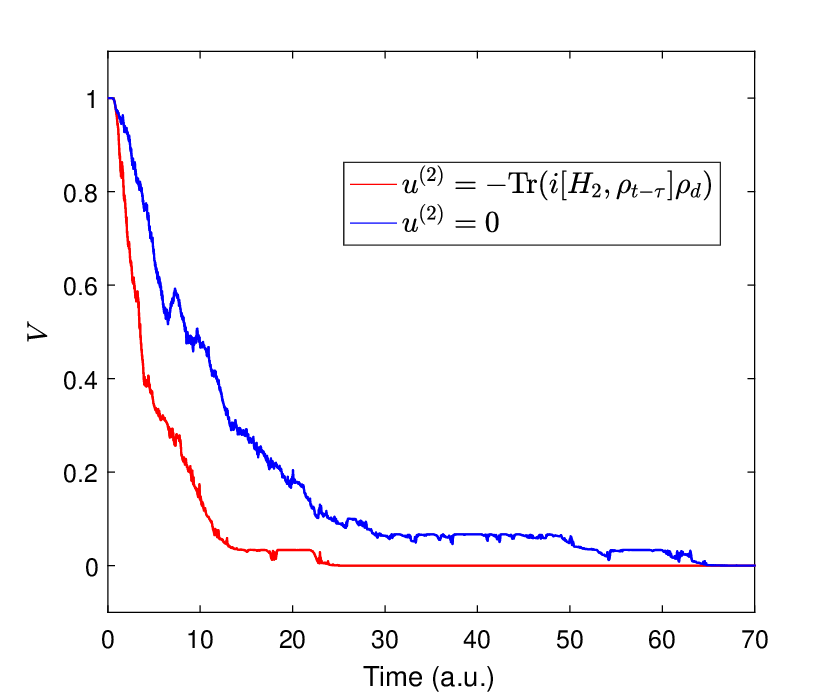}
\caption{Comparison of the distance $V=1-{\rm Tr}^2(\rho \rho_d )$  between $u^{(2)}=0$ and $u^{(2)}=-{\rm Tr}(i[H_2,\rho_{t-\tau}]\rho_d)$ for initial state $\rho_0^{(1)}$ and $\tau=0.1$. \label{Fig.5}}
\end{figure}

\section{Conclusions}
\label{sec6}
This paper investigated the problem of the generation of given entangled states in qubit systems with a constant delay time. Two control strategies have been designed based on the Lyapunov method to drive the system to a desired target state from any initial state. The bang-bang-like control strategy has a simple form with switching between a constant value and zero, while the switching Lyapunov control has the potential to speed up the convergence. The stability of the bang-bang-like control for $N$-qubit systems has been proved theoretically, and we obtained the stability of switching Lyapunov control law in the form of LMIs in a two-qubit example. Numerical results for a two-qubit system show that the proposed control strategies can achieve the control task effectively, and work well for stochastic quantum systems with different delay time. The robustness for these two control strategies in terms of dephasing noises have been tested in simulations, and the results show that the proposed feedback control strategies have great potential to generate and protect quantum entanglement with robust performance. In future work, the strict proof of the switching Lyapunov control for $N$-qubit systems will be considered, and the maximal accepted delay time, and the case with time-varying or uncertain delay time are open questions that are worth further investigating for both two control strategies. 
\appendix
\section{Appendix}
\subsection{Proof of Proposition \ref{proposition1}}
\begin{proof}
The fact that system state converges to the target state $\rho_d$ in probability means:
\begin{equation}
\label{eq19}
{\rm lim}_{t\rightarrow \infty}P \left\{ \|\rho_t-\rho_d \|>\epsilon \right\}=0,\forall \epsilon>0,
\end{equation}
where $\|\cdot\|$ is the Frobenius form.

The continuous function $V(\rho_t)$ converges to $V(\rho_d)$ in probability since the system state $\rho_t$ converges to $\rho_d$ in probability (\cite{R24}, pp.60). Hence, we have:
\begin{equation}
\label{eq20}
{\rm lim}_{t\rightarrow \infty}P \left\{ |V(\rho_t )-V(\rho_d )|>\epsilon \right\}=0,\forall \epsilon >0.
\end{equation}
That means:
\begin{equation}
\label{eq21}
{\rm lim}_{t\rightarrow \infty }P\left\{V(\rho_t )>\epsilon \right\}=0,\forall \epsilon >0.
\end{equation}
As the Lyapunov function satisfies $0\leq V(\rho_t )\leq 1$, we only need to consider the fact of $\epsilon \leq 1$ in \eqref{eq20} and \eqref{eq21}. We calculate $E[V(\rho_t )]$ as
\begin{equation}
\label{eq22}
E[V(\rho_t )]\leq \epsilon \left[1-P\{V(\rho_t )>\epsilon\}\right]+1 \cdot P\{V(\rho_t )>\epsilon\}.
\end{equation}
Using the sign preserving property of the limit, evaluating the limit value of the two sides of \eqref{eq22} and considering \eqref{eq21}, we have
\begin{equation}
\label{eq23}
{\rm lim}_{t\rightarrow \infty}E\left[V(\rho_t )\right]\leq \epsilon,\forall \epsilon > 0.
\end{equation} 
Equation \eqref{eq23} implies:
\begin{equation}
\label{eq24}
{\rm lim}_{t\rightarrow \infty}E\left[V(\rho_t )\right]=0.
\end{equation}
However, $V(\rho_t)$ is continuous and bounded. Using the dominated convergence theorem (\cite{R24}, pp.72), we obtain:
\begin{equation}
\label{eq25}
E\left[{\rm lim}_{t\rightarrow \infty}V(\rho_t ) \right]=0.
\end{equation}
Since the Lyapunov function is non-negative; i.e., $V(\rho_t )\geq 0$, \eqref{eq25} implies:
\begin{equation}
\label{eq26}
{\rm lim}_{t\rightarrow \infty}V(\rho_t )=0.
\end{equation}
That is, $\rho_t$ converges to $\rho_d$ with probability $1$.

On the other hand, $\rho_t$ is a strong Markov process \cite{R12}. According to the convergence of stochastic processes, if the system state converges to $\rho_d$ with probability $1$, it must converge to $\rho_d$ in probability. $\blacksquare$
\end{proof}

%available R9,R33,R34,

\end{document}